\documentclass[a4paper,twocolumn,11pt,accepted=2023-04-13]{quantumarticle}
\pdfoutput=1
\usepackage[utf8]{inputenc}
\usepackage[english]{babel}
\usepackage[T1]{fontenc}
\usepackage{amsmath,amssymb,bbm,amsthm,amsfonts,xcolor}
\usepackage{hyperref}

\usepackage[numbers,sort&compress]{natbib}
\usepackage{tikz}
\usepackage{lipsum}

\newcommand{\Tcal}{\mathcal{T}}
\newcommand{\Ical}{\mathcal{I}}
\newcommand{\Pcal}{\mathcal{P}}
\newcommand{\Pprob}{\mathbb{P}}
\def\tr#1{\mbox{tr}\left[{#1}\right]}

\makeatletter
\theoremstyle{plain}
\newtheorem{thm}{\protect\theoremname}
\theoremstyle{plain}
\newtheorem{lem}{\protect\lemmaname}
\theoremstyle{plain}

\theoremstyle{remark}
\newtheorem*{rem*}{\protect\remarkname}
\theoremstyle{plain}

\theoremstyle{plain}
\newtheorem{cor}{\protect\corollaryname}
\theoremstyle{definition}
\newtheorem{defn}{\protect\definitionname}
\theoremstyle{plain}
\newtheorem{obs}{\protect\observationname}
\theoremstyle{plain}
\newtheorem*{thm*}{\protect\theoremname}
\theoremstyle{plain}
\newtheorem*{lem*}{\protect\lemmaname}
 
\providecommand{\propositionname}{Proposition}
\providecommand{\theoremname}{Theorem}
\providecommand{\lemmaname}{Lemma}
\providecommand{\remarkname}{Remark}
\providecommand{\conjecturename}{Conjecture}
\providecommand{\definitionname}{Definition}
\providecommand{\corollaryname}{Corollary}
\providecommand{\observationname}{Observation}
\allowdisplaybreaks

\def\bra#1{\langle{#1}\vert}
\def\ket#1{\vert{#1}\rangle}
\def\braket#1{\langle{#1}\rangle}

\def\BraVert{e.g.,roup\,\mid\,\bgroup}

\def\ketbra#1#2{\vert{#1}\rangle\!\langle{#2}\vert}

\begin{document}

\title{Hidden Quantum Memory: \\Is Memory There When Somebody Looks?}

\author{Philip Taranto}
\orcid{0000-0002-4247-3901}
\email{philipguy.taranto@phys.s.u-tokyo.ac.jp}
\affiliation{Department of Physics, Graduate School of Science, The University of Tokyo, 7-3-1 Hongo, Bunkyo City, Tokyo 113-0033, Japan}
\affiliation{Atominstitut, Technische Universit{\"a}t Wien, 1020 Vienna, Austria}
\affiliation{Institute for Quantum Optics and Quantum Information, Austrian Academy of Sciences, Boltzmanngasse 3, 1090 Vienna, Austria}

\author{Thomas J. Elliott}
\orcid{0000-0001-5392-9241}
\affiliation{Department of Physics \& Astronomy, University of Manchester, Manchester M13 9PL, United Kingdom}
\affiliation{Department of Mathematics, University of Manchester, Manchester M13 9PL, United Kingdom}

\author{Simon Milz}
\orcid{0000-0002-6987-5513}
\affiliation{School of Physics, Trinity College Dublin, Dublin 2, Ireland}
\affiliation{Institute for Quantum Optics and Quantum Information, Austrian Academy of Sciences, Boltzmanngasse 3, 1090 Vienna, Austria}
\affiliation{Faculty of Physics, University of Vienna, Boltzmanngasse 5, 1090 Vienna, Austria}

\maketitle

\begin{abstract}
In classical physics, memoryless dynamics and Markovian statistics are one and the same. This is not true for quantum dynamics, first and foremost because quantum measurements are invasive. Going beyond measurement invasiveness, here we derive a novel distinction between classical and quantum processes, namely the possibility of \textit{hidden quantum memory}. While Markovian statistics of classical processes can always be reproduced by a memoryless dynamical model, our main result shows that this is not true in quantum mechanics: We first provide an example of quantum non-Markovianity whose manifestation depends on whether or not a previous measurement is performed---an impossible phenomenon for memoryless dynamics; we then strengthen this result by demonstrating statistics that are Markovian independent of how they are probed, but are nonetheless \textit{still} incompatible with memoryless quantum dynamics. Thus, we establish the existence of Markovian statistics gathered by probing a quantum process that nevertheless \textit{fundamentally} require memory for their creation.
\end{abstract}

Our ability to understand and control memory effects in the evolution of open quantum systems is becoming increasingly important as technology allows us to manipulate interactions with increasing levels of speed, precision and complexity~\cite{Preskill_2018,Acin_2018}. Control over memory can be advantageous in various tasks, such as creating, manipulating and preserving coherences and correlations~\cite{Banaszek_2004,Barreiro_2010}, simulating complex dynamics~\cite{Sayrin_2011,Grimsmo_2015,Iles-Smith_2016,Cerrillo_2016,Metelmann_2017,Whalen_2017,Basilewitsch_2017,Fischer_2019,Luchnikov_2019,Jorgensen_2019,Liu_2019,Jorgensen_2020,Magrini_2021,Elliott_2022}, implementing randomised benchmarking and error correction~\cite{Ball_2016,Figueroa_2021PRX,Figueroa_2022}, performing optimal dynamical decoupling~\cite{Viola_1999,Biercuk_2009,Addis_2015}, designing quantum circuit architectures~\cite{Chiribella_2008_Quantum,Chiribella_2009,Mavadia_2018,White_2020,White_2021,Taranto_2021_Exp,White_2021-2}, and improving the efficiency of thermodynamic machines~\cite{Bylicka_2016,Kato_2016,Taranto_2020,Taranto_2023}. 

One has no choice but to account for complex noise and memory effects when modelling realistic dynamical systems, as no system is truly isolated; in general, the environment stores information about the past and propagates it in time, leading to memory effects that manifest themselves as complex multi-time correlations~\cite{Rivas_2014,Breuer_2016,deVega_2017,Li_2018,TarantoThesis}. A special case of open dynamics are memoryless dynamics, for which the environment retains \textit{no} memory of its previous interactions with the system. Such dynamics have been studied extensively due to their accuracy in many practically relevant situations and their exponentially reduced complexity from the general scenario. Both in the classical and quantum setting, such efficient descriptions arise by way of (time-local) master equations that efficiently simulate the system dynamics~\cite{Lindblad_1976,Gorini_1976,Manzano_2020}; in practice, the assumption of memorylessness is often made for simplicity and describes many `real-world' scenarios with a high degree of accuracy~\cite{Carmichael,Breuer,vanKampen_2011,Wio_2012}. 

However, experimentally determining that a quantum process is memoryless requires full process tomography, which necessitates a myriad of complex sequential measurements and has consequently only been done for low-dimensional/few timestep cases~\cite{ringbauer_characterizing_2015, White_2020,Taranto_2021_Exp,White_2021}. A more tractable situation is the sequential probing of a fixed observable via sharp, projective measurements. In this case, memoryless quantum processes---like their classical counterparts---lead to Markovian statistics, i.e., statistics where the future is conditionally independent of the past. Thus, at first glance, memorylessness of the dynamics seems to manifest on the experimental level identically for classical and quantum processes. However, this is not the case; for one, quantum measurements of any observable are generally invasive, leading to inconsistent \mbox{(sub-)statistics}~\cite{Milz_2017_Kolmogorov} and the violation of Leggett-Garg inequalities~\cite{Leggett_1985,Leggett_2008,Emary_2014}. In contradistinction, measuring an observable in the classical world can be done non-invasively. Moreover---beyond measurement invasiveness---here we demonstrate that quantum processes can yield Markovian statistics \textit{that fundamentally require memory} for their creation. 

More concretely, in classical physics, \textit{any} Markovian statistics can be described by a memoryless dynamical model (i.e., as emerging from a sequence of independent stochastic matrices). In the quantum case, measuring a fixed observable no longer constitutes a tomographically complete procedure; consequently, the existence of processes with memory that nonetheless lead to Markovian statistics when said observable is probed is not surprising \textit{per se} and has been demonstrated~\cite{Taranto_2019L,Taranto_2019A,Taranto_2021S}. This phenomenon notwithstanding, for any (quantum) experiment that yields Markovian statistics, it is reasonable to believe that there \textit{always exists} some memoryless quantum dynamics that faithfully reproduces the observed statistics. Such a description is known as the \emph{quantum regression formula} \textbf{(QRF)}~\cite{Lax_1963,Carmichael,Breuer} and is a widely used assumption that links operational quantities---namely, recorded statistics---to dynamical ones---namely, a model of the underlying dynamics. 

Here, we ask the question: \emph{Can Markovian statistics always be faithfully reproduced by a memoryless dynamical model?} In other words, can the QRF \emph{always} be employed to describe Markovian statistics? Our main result, perhaps surprisingly, answers this in the negative. Since this contradicts the counterpart answer within classical physics (i.e., for sharp measurements of a given observable), we thus uncover a new type of genuinely quantum phenomenon: \emph{Hidden quantum memory}. This observation makes quantum memory an emergent phenomenon: Observing Markovianity with respect to a fixed measurement basis is not sufficient to guarantee the existence of a memoryless dynamical descriptor. Such hidden quantum memory is similar in spirit to other quantum traits that require precisely the resource in their implementation that they ultimately hide, such as quantum channels that preserve all separable states but cannot be implemented via local operations and classical communication~\cite{Plenio_2007,Horodecki_2009,Chitambar_2020}, non-signalling maps that require signalling~\cite{beckman_causal_2001}, and maximally incoherent operations that necessitate coherent resources~\cite{Chitambar_2016_Critical,Chitambar_2016_Comparison,Marvian_2016}. We begin by outlining the envisaged setup before detailing key properties of memoryless dynamics (both classical and quantum).

\section{Framework}

In any experimental scheme concerning temporal processes, an experimenter probes a system of interest at (any subset of) times $\mathcal{T}_n := \{ t_1, \hdots, t_n \}$ (with $t_n > \dots > t_1$) and records the corresponding probability distributions $\{\Pprob(\boldsymbol{x}_\Gamma)\}$, where $\Gamma \subseteq \Tcal_n$ and $\boldsymbol{x}_\Gamma:= \{x_j|t_j\in \Gamma\}$ (see Fig.~\ref{fig::StochProc}). These capture, for instance, the probability that $x_1$ is observed at time $t_1$ \emph{and} $x_2$ at $t_2$, and so on, with all possible combinations of measurement times. Note that the experimenter can also \emph{not} make a measurement at any intermediate time, e.g., record $\mathbbm{P}(x_3, x_1)$ without measuring at $t_2$. 

\begin{figure}
    \centering
    \includegraphics[width=0.95\linewidth]{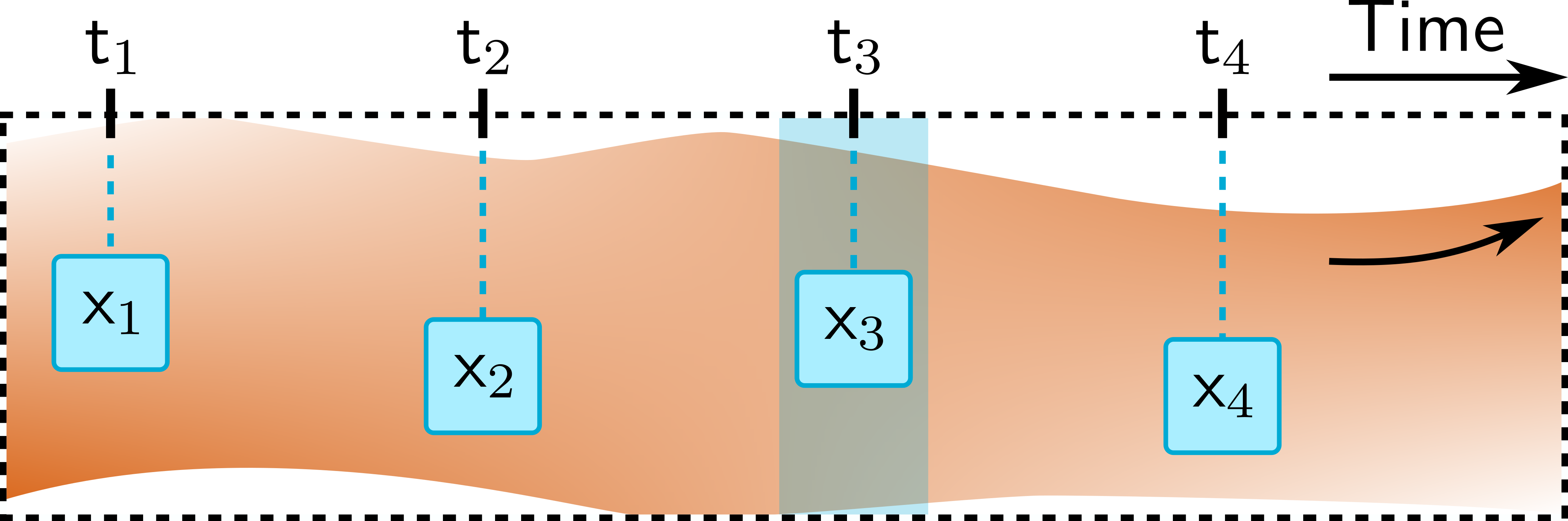}
    \caption{\textit{Probing dynamics.} By probing a process---e.g., Brownian motion, or the evolution of a spin degree of freedom---sequentially (here, at times $\Tcal_4=\{t_1,t_2,t_3,t_4\}$), an experimenter can deduce the probability distribution $\Pprob(x_4, x_3,x_2,x_1)$. In the classical case, this also includes all `contained' distributions, e.g., $\Pprob(x_4,x_2,x_1)$ via marginalisation [see Eq.~\eqref{eq:qrf-classicalmarginalisation}]. In the quantum case, due to invasiveness, deducing said distributions requires a new experiment where \emph{no} measurement is performed at $t_3$ (depicted by the shaded box).}
    \label{fig::StochProc}
\end{figure}

Independent of the physical scenario---it could be classical, quantum, or even post-quantum---one can define the concept of \emph{Markovianity} based on the observed statistics alone, as conditional independence of any current outcome from all but the most recent one. Concretely, we have the following working definition: 
\begin{defn}
\label{def::Markovian}
A \emph{Markovian statistics} on a set of times $\Tcal_n$ is a collection of conditional probability distributions $\{\Pprob(x_j|x_{j-1},\dots, x_1)\}_{t_j \in \Tcal_n}$ for which
\begin{gather}
   \Pprob(x_j|x_{j-1},\dots, x_1) = \mathbbm{P}(x_j|x_{j-1})
\end{gather}
for all $t_j \in \Tcal_n$.
\end{defn}
Defined as such, the question of Markovianity is, a priori, \emph{theory independent} and concerns only the observed statistics. As we shall see, though, the concept of conditional probabilities is a subtle issue that depends on the envisaged physical scenario. Throughout this article, we distinguish Markovianity from the notion of \emph{memoryless dynamics}, which corresponds to the memory properties of the \emph{underlying} dynamics that engenders the observed statistics, thereby making the latter a \emph{theory-dependent} concept.

Specifically, the question of memorylessness concerns whether, throughout the evolution of a system that is coupled to some inaccessible environment, said environment perpetuates past information about the system forward in time or irretrievably dissipates it.\footnote{We consider memory to be a property that is external to the system, i.e., stored in the environment, rather than information encoded in the system itself.} The description of such open evolution differs across physical theories: In the classical setting, the most general state transformations are stochastic matrices, whereas in the quantum realm, these are quantum channels. Probability distributions arising from interrogating either classical or quantum processes therefore have different properties since they are calculated via different rules. Our main result shows that such a distinction holds for the relationship between Markovianity and memorylessness: Although equivalent in the classical case, in the quantum realm the observation of Markovian statistics does \emph{not} guarantee even the existence of a memoryless dynamics that engenders them. 

\section{Classical Dynamics}

We begin with a discussion of memoryless \textit{classical} dynamics:
\begin{defn}\label{def:qrf-classicalmemoryless}
A \emph{memoryless classical dynamics} on $\mathcal{T}_n$ is a set of mutually independent stochastic matrices $\{ S_{j:j-1}\}_{j=2,\hdots,n}$ and an initial state (i.e., probability vector) $\boldsymbol{p}
_1$ such that the probability distribution over any sequence of outcomes $x_1,\hdots,x_n$ is given by
\begin{align}\label{eq:qrf-classicalregressionformula}
    &\mathbbm{P}(x_n,\hdots,x_1) = \notag \\*
    &\bra{x_n} S_{n:n-1} \ketbra{x_{n-1}}{x_{n-1}} \hdots \ketbra{x_2}{x_2}S_{2:1}\ketbra{x_1}{x_1} \boldsymbol{p}
_1,
\end{align}
where $\ketbra{x_j}{x_j}$ are projectors corresponding to measurement outcomes $x_j$.
\end{defn}

Although the environment plays a role in dictating the state transitions between any times $t_{j-1}$ and $t_{j}$---namely via the stochastic matrices $S_{j:j-1}$, which are matrices with non-negative elements whose columns sum to unity---in \emph{memoryless} processes, the environment does not propagate information, i.e., the stochastic matrices in Eq.~\eqref{eq:qrf-classicalregressionformula} are mutually \emph{independent}. On the other hand, Markovianity (see Def.~\ref{def::Markovian}) concerns only the observed statistics [l.h.s. of Eq.~\eqref{eq:qrf-classicalregressionformula}]. In classical physics, we can make the following simple observation (see App.~\ref{app:qrf-classical}):

\begin{obs}\label{obs:qrf-classicalmemorylessimpliesmarkov}
In the classical setting, memoryless dynamics are equivalent to Markovian statistics.
\end{obs}
Specifically, this equivalence is given by setting $\bra{x_j}S_{j:j-1}\ket{x_{j-1}} = \mathbbm{P}(x_j|x_{j-1})$, and it follows from Eq.~\eqref{eq:qrf-classicalregressionformula} that for any Markovian statistics/memoryless classical dynamics we have
\begin{align}\label{eq:qrf-classicalmarkov}
    \mathbbm{P}(x_n,\hdots, x_1) =& \mathbbm{P}(x_n|x_{n-1})  \mathbbm{P}(x_{n-1}|x_{n-2}) \, \hdots \notag \\*
    &\times \hdots \, \mathbbm{P}(x_2|x_1) \, \mathbbm{P}(x_1).
\end{align}
In one direction, Obs.~\ref{obs:qrf-classicalmemorylessimpliesmarkov} states that for any memoryless dynamics, the observed statistics are Markovian---this is also true in the quantum setting (see below). Conversely, if one records Markovian statistics by probing a classical process, then one can always construct a unique, memoryless dynamical model that faithfully reproduces them---as we will see, this is not true for statistics gathered from quantum processes. 

A major distinction between classical and quantum processes (memoryless or not) is that in the classical realm, the single $n$-time probability distribution $\mathbbm{P}(x_n, \hdots ,x_1)$ contains the entire set of statistics on all subsets of times $\Gamma\subseteq \Tcal_n$. 
That is, the probability $\Pprob(\boldsymbol{x}_{\Gamma}, \Ical_{\overline \Gamma})$ to observe a sequence of outcomes $\boldsymbol{x}_{\Gamma}$ when probing the process at times $\Gamma$ and \emph{not} measuring (denoted by the `do-nothing' instrument $\Ical_{\overline\Gamma}$) at the remaining times $\overline \Gamma := \Tcal_n\setminus \Gamma$ can be deduced via marginalisation 
\begin{gather}
    \label{eq:qrf-classicalmarginalisation}
    \mathbbm{P}(\boldsymbol{x}_{\Gamma},\mathcal{I}_{\overline \Gamma}) = \sum_{\boldsymbol{x}_{\overline \Gamma}} \mathbbm{P}(x_n, \hdots, x_1),
\end{gather}
This \emph{non-invasiveness} of measurements in classical physics underlies Obs.~\ref{obs:qrf-classicalmemorylessimpliesmarkov} and similarly fails to hold in quantum mechanics. As a direct consequence of measurement non-invasiveness, the properties of a memoryless classical process on $\Tcal_n$ translate to all `sub-processes' that are probed only at times $\Gamma \subset \Tcal_n$ (see App.~\ref{app:qrf-classical}):
\begin{cor}\label{cor:qrf-classicalmarkovsubstatistics}
All sub-statistics of a memoryless classical dynamics are Markovian and the corresponding conditional probabilities are compatible.
\end{cor}
By \emph{compatible}, we mean that all conditional probabilities are independent of how they are obtained, i.e.,
\begin{gather}
\label{eqn::Compatibility}
    \frac{\Pprob(x_j,\boldsymbol{x}_{\Gamma^{(i)}})}{\Pprob(\boldsymbol{x}_{\Gamma^{(i)}})} =  \frac{\Pprob(x_j,\boldsymbol{x}_{\Gamma^{(i)\prime}})}{\Pprob(\boldsymbol{x}_{\Gamma^{(i)\prime}})} =: \Pprob(x_j|x_i), 
\end{gather}
for all $t_j, t_i \in \Tcal_n$ (with $t_j>t_i$) and all subsets $\Gamma^{(i)},\Gamma^{(i)\prime} \subseteq \Tcal_n$ that contain $t_i$ as their largest time. For a classical memoryless dynamics, knowledge of any outcome $x_i$ suffices to erase all historic information (including whether or not a previous measurement was made) and is therefore the only relevant parameter for predicting future outcomes. Such compatibility between Markovian sub-statistics of a memoryless quantum dynamics also holds (for sharp measurements of a fixed observable), although it is less obvious, and we will later employ the breakdown of compatibility as a witness for memory. 

\section{Quantum Dynamics}

In contrast to classical physics, in quantum mechanics, measurements are generally \textit{invasive} such that there is a difference between averaging over outcomes and not performing a measurement. In this article, we focus on the generally considered situation of sharp measurements of an observable, e.g., position in the classical case or spin in the quantum case. This allows us to fairly compare `classical' and `quantum' processes in time. Within this setting, the measurements themselves do not `actively' change the state of the observed system (in the sense that no active interventions are performed), and thus measurement invasiveness only manifests itself in quantum mechanics (due to the loss of coherences in the observed state).

Subsequently, this makes (conditional) probabilities in the quantum realm protocol-dependent entities that require further specification; in what follows, whenever we consider a probability distribution $\Pprob(\boldsymbol{x}_\Gamma)$, we mean the statistics obtained from \textit{only} performing measurements at times in the set $\Gamma \subseteq \Tcal_n$, and doing nothing (denoted by $\Ical_{\overline \Gamma}$) at the remaining times $\overline \Gamma = \Tcal_n \setminus \Gamma$. Importantly, in quantum mechanics---in contrast to Eq.~\eqref{eq:qrf-classicalmarginalisation}---$\Pprob(\boldsymbol{x}_\Gamma):=\mathbbm{P}(\boldsymbol{x}_{\Gamma},\mathcal{I}_{\overline \Gamma}) \neq \sum_{\boldsymbol{x}_{\overline \Gamma}} \mathbbm{P}(x_n, \hdots, x_1)$. Such measurement invasiveness is well-studied and has been used to witness the non-classicality of physical processes~\cite{Milz_2017_Kolmogorov, Smirne_2019_Coherence,Strasberg_2019,Milz_2020}. Despite these added subtleties in the definition of (conditional) probabilities, memoryless quantum dynamics lead---just like in the classical case---to well-defined, compatible Markovian statistics and sub-statistics. To see this, we first generalise Def.~\ref{def:qrf-classicalmemoryless} to the quantum case: 
\begin{defn}\label{def:qrf-memorylessquantum}
A \emph{memoryless quantum dynamics} on $\Tcal_n$ is a set of mutually independent completely positive and trace preserving \textbf{(CPTP)} maps $\{ \Lambda_{j:j-1}\}_{j=2,\hdots,n}$ and an initial state (density operator) $\rho_1$ such that the probability distribution over any sequence of outcomes $x_1,\hdots,x_n$ is given by
\begin{align}\label{eq:qrf-quantumregressionformula} 
    \mathbbm{P}(x_n,\hdots,x_1) = \tr{\Pcal_n^{(x_n)} \Lambda_{n:n-1}  \hdots \Lambda_{2:1} \Pcal_1^{(x_1)} [\rho_1]},
\end{align}
where $\mathcal{P}_j^{(x_j)}[\,\bullet\,] := \ketbra{x_j}{x_j} \bullet \ketbra{x_j}{x_j}$ are maps corresponding to sharp (i.e., rank-1) projective measurements.
\end{defn}
Analogous to the classical case, CPTP maps are the most general state transformations in the presence of environmental noise, and the absence of memory in the dynamics corresponds to the mutual independence of the maps $\Lambda_{j:j-1}$ in the definition. The above equation to compute probabilities is commonly known as the \emph{quantum regression formula} \textbf{(QRF)}~\cite{Lax_1963,Carmichael,Breuer}. Importantly, it allows for the computation of sub-statistics on any $\Gamma \subseteq \Tcal_n$, not via marginalisation, but by replacing the projection operators corresponding to probing times in $\overline \Gamma$ in Eq.~\eqref{eq:qrf-quantumregressionformula} with identity maps. Of course, one need not perform projective measurements, and the above formula can be used to calculate the probability distribution over \textit{any} sequence of outcomes for arbitrary instruments. However, in contrast to the non-invasive measurements typically considered in classical stochastic processes, such general quantum measurements do not necessarily reset the state of the system, which means that memoryless quantum dynamics can lead to non-Markovian statistics for general instruments~\cite{Pollock_2018L,Pollock_2018A,Taranto_2019L,Taranto_2019A}. Nonetheless, when restricted to sharp, projective measurements of a given observable, then---just as in the classical setting---memorylessness in the quantum realm manifests itself on the observational level as Markovianity (see App.~\ref{app:qrf-quantum}): 
\begin{lem}\label{lem:qrf-quantummemorylessimpliesmarkovianity}
    Any memoryless quantum dynamics leads to Markovian statistics (for sharp, projective measurements).
\end{lem}

We saw earlier that memoryless classical processes also lead to (compatible) Markovian \emph{sub}-statistics (see Cor.~\ref{cor:qrf-classicalmarkovsubstatistics}), where compatibility is given by Eq.~\eqref{eqn::Compatibility}. This is also true for memoryless quantum processes, with the important difference that sub-statistics are not obtained by marginalisation, but by `doing nothing' at the excluded times, i.e., by explicitly performing the experiment in a different way. Probing sub-statistics in this manner yields meaningful conditional probabilities and we have the following (see App.~\ref{app:qrf-quantum}): 
\begin{lem}\label{lem:qrf-quantummemorylessimpliesmarkoviansubstatistics}
    Any memoryless quantum dynamics leads to Markovian sub-statistics (for sharp, projective measurements) that are compatible.
\end{lem}

In both quantum mechanics and classical physics, memoryless dynamics---when probed sharply in a fixed basis---\emph{always} lead to Markovian statistics and Markovian, compatible sub-statistics. In the classical setting, the converse is also true: From the observation of Markovian statistics one can always construct a (unique) memoryless process describing the situation at hand. As discussed, measuring a fixed observable of a time-evolving quantum system cannot provide enough information to fully determine the underlying dynamics. Nonetheless, it is reasonable to assume that whenever one observes Markovian statistics, there should exist \emph{some} memoryless description that correctly reproduces them (indeed, this is the assumption of employing the QRF to describe Markovian statistics). Thus, we now ask the  question: \emph{Given Markovian statistics (deduced via sharp, projective measurements), does there always exist a memoryless quantum dynamical model that faithfully reproduces them?}  

\section{Hidden Quantum Memory \& Incompatability}

We answer the above question in the negative, first by demonstrating a quantum process that leads to Markovian statistics with non-Markovian sub-statistics, and then by constructing a process with Markovian statistics and sub-statistics that are nonetheless incompatible.

\begin{thm}\label{thm:qrf-hiddennonmarkovianity}
    Given Markovian statistics on $\Tcal_n$ (deduced via sharp, projective measurements), there does \emph{not} always exist a memoryless quantum dynamics that faithfully reproduces them.
\end{thm}

\begin{figure}
    \centering
    \includegraphics[width=0.95\linewidth]{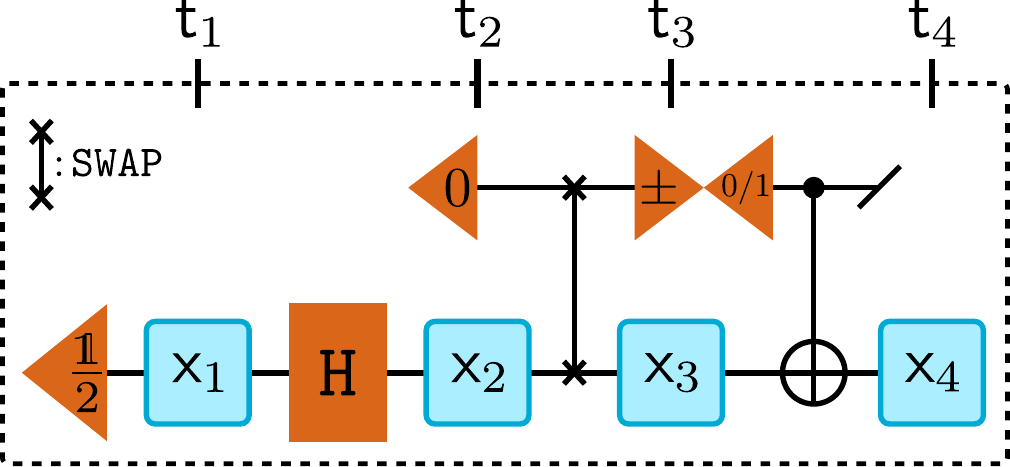}
    \caption{\textit{Markovian statistics that require memory.} When the $\sigma_z$ observable is measured (sharply) at \textit{all} times, the circuit yields Markovian statistics. Memory becomes apparent in the joint statistics $\mathbbm{P}(x_4,x_3,\mathcal{I}_2,x_1)$ when no measurement is performed at $t_2$, which is in contradiction with the possibility of a memoryless dynamical model (Thm.~\ref{thm:qrf-hiddennonmarkovianity}).}
    \label{fig::Circuit}
\end{figure}

\begin{proof}
Our proof is by way of constructing an explicit example, depicted in Fig.~\ref{fig::Circuit}. The dynamics is over four times and the experimenter always measures the $\sigma_z$ observable. An initial state $\rho_1=\tfrac{\mathbbm{1}}{2}$ is sent to the experimenter, who measures it. The dynamics between times $t_1$ and $t_2$ is a Hadamard gate. Following the measurement at $t_2$, the system is swapped with a fiducial environment state  $\tau = \ket{0}$, which is what the experimenter measures at time $t_3$. Meanwhile, the dynamics of the environment consists of a measurement in the $\sigma_x$-basis, followed by a preparation of $\ket{0} (\ket{1})$ whenever $+ (-)$ is recorded. Between times $t_3$ and $t_4$, the dynamics comprises a \verb!CNOT! gate, controlled on the environment. Due to the gates that act on the system \textit{and} environment, this circuit can, in principle, display memory effects for the system dynamics. In App.~\ref{app:qrf-hiddenquantummemory}, we calculate the full statistics $\mathbbm{P}(x_4,x_3,x_2,x_1)$ and show them to be Markovian, i.e., $\Pprob(x_4|x_3,x_2,x_1) = \Pprob(x_4|x_3)$ and $\Pprob(x_3|x_2,x_1) = \Pprob(x_3|x_2)$. This is because the measurement of $\sigma_z$ at $t_2$ yields an output state that is unbiased with respect to the $\sigma_x$-basis measurement on the environment and therefore all memory of $x_1$ is lost. However, by calculating the sub-statistics where the experimenter does \emph{not} measure at time $t_2$, i.e., $\mathbbm{P}(x_4,x_3,\Ical_2,x_1)$, we see that they are \emph{non}-Markovian since information about $x_1$ is now \emph{not} fully scrambled by the `intervention' (or rather lack thereof) at $t_2$, and we have $\mathbbm{P}(x_4|x_3,\Ical_2,x_1) \neq \mathbbm{P}(x_4|x_3)$ with dependence on $x_1$. As we proved in Lem.~\ref{lem:qrf-quantummemorylessimpliesmarkoviansubstatistics}, such behaviour cannot happen for \emph{any} memoryless quantum dynamics. Thus, even though the statistics on $\Tcal_n$ is Markovian, there is \emph{no} memoryless quantum dynamics that faithfully reproduces the statistics on all four times, since the sub-statistics fail to be Markovian.
\end{proof}

Here, we have uncovered a new temporal quantum phenomenon: \emph{Hidden quantum memory}. The fact that full statistics can be Markovian but sub-statistics can be non-Markovian for measurements of a given observable is impossible in the classical realm. Moreover, this property cannot occur for memoryless quantum dynamics either (whenever said observable is measured sharply). Thus, we have shown the existence of Markovian statistics that, not only potentially come from a quantum dynamics with memory (which can happen, as is well known, when measured in a fixed basis), but \emph{fundamentally require} memory for their reproduction.

Another way of viewing this result is that non-Markovian sub-statistics serves as a witness for the necessity of memory in the underlying quantum dynamics. This naturally begs the question: \emph{If the full statistics and all sub-statistics are Markovian, does there always exist a memoryless quantum dynamical model that faithfully reproduces them}? In other words, is the ability to detect non-Markovian sub-statistics a requirement for ruling out a memoryless description of the quantum dynamics? Here, we also answer this in the negative, providing an even stronger result than above:
\begin{thm}\label{thm:qrf-quantummarkovincompatible}
    Given Markovian statistics and sub-statistics on $\Tcal_n$ and all subsets thereof (deduced via sharp, projective measurements), there does not always exist a memoryless quantum process that faithfully reproduces them.
\end{thm}
\begin{proof}
The proof is again by constructing an explicit example, with the corresponding circuit depicted in Fig.~\ref{fig::Circuit2}. In App.~\ref{app:qrf-hiddenquantummemory}, we calculate the full statistics $\mathbbm{P}(x_4,x_3,x_2,x_1)$ and all relevant sub-statistics [e.g., $\mathbbm{P}(x_4,x_3,\Ical_2,x_1)$, etc.] and show them to be Markovian. This latter fact can easily be seen directly: Since the state of the system is discarded and reprepared in a fixed state $\ket{0}$ between times $t_2$ and $t_3$, the only way in which memory from $t_1$ and/or $t_2$ can influence the statistics observed at $t_4$---thus potentially rendering the conditional probabilities $\Pprob(x_4|x_3, x_2, x_1)$, $\Pprob(x_4|x_3, \Ical_2, x_1)$ and $\Pprob(x_4|x_3, x_2, \Ical_1)$ non-Markovian---is via the state of the environment at time $t_3$. However, while this state explicitly depends on \textit{whether} previous measurements were performed, it crucially does \textit{not} depend on the respective measurement outcomes. For example, if measurements at $t_1$ \textit{and} $t_2$ are performed, then the state of the environment at $t_3$ is proportional to the maximally mixed state, independent of the respective measurement outcomes. On the other hand, if only a measurement at $t_1$ is performed, then the state of the environment at $t_3$ is (proportional to) $\ketbra{0}{0}$, again independent of the measurement outcome at $t_1$. The same independence of previous outcomes (but not of whether or not the respective measurements were performed) holds true for all other potential combinations of performed and unperformed measurements (as we show explicitly in App.~\ref{app:qrf-hiddenquantummemory}), such that \textit{all} conditional probabilities observed at times $t_4$ and $t_3$ are indeed Markovian. However, as can already be seen from the above discussion, they are not compatible. Indeed, since the state of the environment at time $t_3$ depends upon whether or not a measurement was performed at $t_2$, the resulting conditional probabilities at $t_4$ differ depending upon whether or not previous measurements were performed, making them incompatible.
\end{proof}

\begin{figure}
    \centering
    \includegraphics[width=0.95\linewidth]{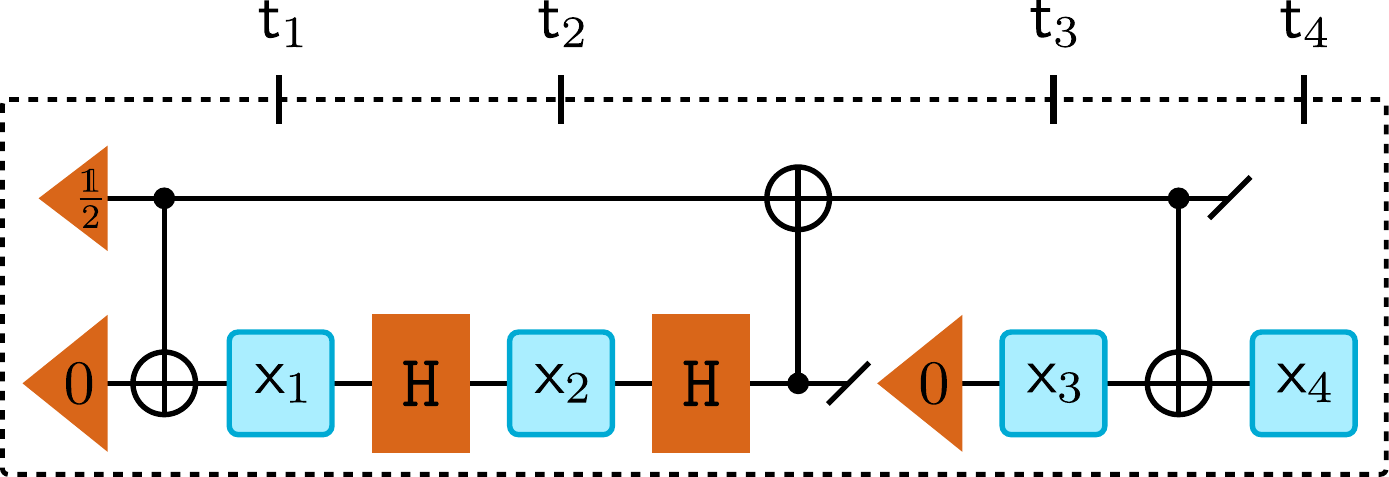}
    \caption{\textit{Incompatible Markovian statistics.} When the $\sigma_z$ observable is measured on the system at all subsets of times, the circuit yields Markovian statistics and sub-statistics. Despite this Markovianity, the respective conditional probabilities are incompatible, i.e., they depend on whether or not previous measurements were performed (Thm.~\ref{thm:qrf-quantummarkovincompatible}). Such behaviour is \emph{only} possible in the presence of memory.}
    \label{fig::Circuit2}
\end{figure}

\section{Conclusions}

In this article, we have presented the concept of hidden quantum memory, i.e., the existence of processes that yield Markovian statistics (for sharp measurements of an observable) that \emph{cannot} be explained without underlying memory. In a similar vein to the violation of Leggett-Garg inequalities, this phenomenon can only occur when the performed measurements are invasive, since otherwise the observed statistics are classical and hence Markovianity and memorylessness coincide. However, hidden quantum memory is not merely a different manifestation of measurement invasiveness, but arises due to its interplay with memory effects; while memoryless quantum dynamics can violate Leggett-Garg inequalities, they cannot propagate information about whether or not measurements were performed at earlier times other than the most recent---i.e., they cannot exhibit hidden quantum memory.

It is important to stress the assumptions that underpin our observations, in particular with respect to the condition that the measurements of the observable are sharp. Indeed, allowing for classical stochastic processes to be probed via \textit{active} interventions~\cite{Pearl,Costa2016} or noisy measurements~\cite{Siefert2003,Bottcher2006,Kleinhans2007,Lehle2011} can also lead to seemingly `non-classical' effects such as the violation of Leggett-Garg inequalities~\cite{Leggett_1985,Leggett_2008,Emary_2014}, breakdown of Kolmogorov consistency~\cite{Milz_2017_Kolmogorov}, or the device-dependence of memory~\cite{Taranto_2019L,Taranto_2019A}, and even the notion of Markovianity itself becomes either obfuscated or trivial~\cite{Capela2022}. In fact, with such interventions the same statistics as generated by the dynamics of Figs. (2) and (3) can be reproduced with classical dynamics (with memory). However, in the classical case, allowing for active manipulations of the state to permit such incompatible statistics is rather ad hoc---essentially requiring the observer to explicitly communicate to the environment whether or not a measurement is made---and not intrinsic to the fundamental properties of measurements themselves. This is in contrast to the situation in quantum mechanics, where measurements fundamentally disturb the system state in general. By only considering sharp measurements of an observable, we restrict our focus to the properties of sequential measurements per se and fairly compare the two theories at the expense of rendering our results device-dependent (which is necessary for any meaningful distinction between classical and quantum temporal effects~\cite{Budroni_2019}). Within this paradigm, then, activation of hidden memory can only be seen for quantum dynamics, making the phenomena we uncover genuinely quantum effects.

Importantly, our results differ from the (known) fact that probing a quantum processes with memory in a fixed basis---i.e., sharply measuring a fixed observable---can yield Markovian statistics. For the Markovian statistics that we reported, there exists \textit{no} memoryless model that reproduces them, either because they become non-Markovian when measurements are not performed at some times, or because all observed statistics and sub-statistics are Markovian but incompatible. In turn, this implies that even if one observes Markovian statistics in a given basis, one cannot confidently employ a QRF to compute the statistics in said basis. As a consequence, even detecting the possibility of a memoryless description of a process is an experimentally complex undertaking that not only requires one to deduce joint probabilities on $\Tcal_n$, but also on all subsets thereof. Naturally, one might expect that simultaneously demanding Markovianity \textit{and} compatibility of all observed sub-statistics should suffice to guarantee a memoryless description. However, even under such strong requirements, the existence of a memoryless model is a priori not clear, and investigations into this question are subject to future work. 

Together, our results expose a novel genuinely quantum effect in time and demonstrate the richness of effects that arise from the intricate interplay of measurement invasiveness, memory, and the freedom to choose different instruments that quantum mechanics affords.

\vspace{-0.5em}
\begin{acknowledgments}
We would like to thank Dario Egloff, Andrea Smirne, Kavan Modi, and Top Notoh for interesting discussions. P.T. acknowledges funding from the Japan Society for the Promotion of Science (JSPS) Postdoctoral Fellowships for Research in Japan (Short-term Program) and the KAKENHI grant No. 21H03394, Austrian Science Fund (FWF) project: Y879-N27 (START), and the European Research Council (Consolidator grant `Cocoquest' 101043705). T.J.E. is supported by the University of Manchester Dame Kathleen Ollerenshaw Fellowship. S.M. acknowledges funding from the Austrian Science Fund (FWF): ZK3 (Zukunftkolleg) and Y879-N27 (START project), the European Union's Horizon 2020 research and innovation programme under the Marie Sk{\l}odowska Curie grant agreement No. 801110, and the Austrian Federal Ministry of Education, Science and Research (BMBWF). The opinions expressed in this publication are those of the authors, the EU Agency is not responsible for any use that may be made of the information it contains. This project/research was supported by grant number FQXi-RFP-IPW-1910 from the Foundational Questions Institute and Fetzer Franklin Fund, a donor advised fund of Silicon Valley Community Foundation.
\end{acknowledgments}

\vspace{-0.5em}
\def\bibsection{\section*{References}} 

%

\onecolumn\newpage
\appendix

\section{Classical Dynamics}\label{app:qrf-classical}

\subsection{Memoryless Classical Dynamics and Markovian Statistics}
Here, we prove Obs.~\ref{obs:qrf-classicalmemorylessimpliesmarkov} of the main text: \\ 

\noindent \textbf{Observation~\ref{obs:qrf-classicalmemorylessimpliesmarkov}.} \textit{In the classical setting, memoryless dynamics are equivalent to Markovian statistics.} \\ 

Naturally, this equivalence is well-known, but its explicit discussion exposes many of the subtleties with respect to marginalisation that play a crucial role in the quantum case. For the proof, in the forwards direction, beginning with Eq.~\eqref{eq:qrf-classicalregressionformula}, we have
\begin{align}\label{eq:qrfapp-classicalregressionformula}
    \mathbbm{P}(x_n,\hdots,x_1) = \bra{x_n} S_{n:n-1} \ketbra{x_{n-1}}{x_{n-1}} \hdots \ketbra{x_2}{x_2}S_{2:1}\ketbra{x_1}{x_1}\boldsymbol{p}
_1\, , 
\end{align}
where $\{S_{j:j-1}\}$ and $\boldsymbol{p}_1$ are, respectively, the stochastic matrices and the initial probability vector that define the memoryless classical dynamics. Computing the conditional probability for an arbitrary time $t_j$ given the entire sequence of historic outcomes up until that time explicitly gives
\begin{align}
    \mathbbm{P}(x_j|x_{j-1},\hdots,x_1) &= \frac{\bra{x_j} S_{j:j-1} \ketbra{x_{j-1}}{x_{j-1}} S_{j-1:j-2} \ketbra{x_{j-2}}{x_{j-2}} \hdots \ketbra{x_2}{x_2}S_{2:1}\ketbra{x_1}{x_1}\boldsymbol{p_1}}{\bra{x_{j-1}} S_{j-1:j-2} \ketbra{x_{j-2}}{x_{j-2}} \hdots \ketbra{x_2}{x_2}S_{2:1}\ketbra{x_1}{x_1} \boldsymbol{p_1}} \notag \\
    &= \bra{x_j} S_{j:j-1} \ket{x_{j-1}} \quad \quad \forall \, x_{j-2},\hdots,x_1.
\end{align}
This expression is independent of all $x_1,\hdots,x_{j-2}$ and it is indeed equivalent to the 
the conditional probability $\mathbbm{P}(x_j|x_{j-1})$ an experimenter would observe when only making measurements at $t_{j-1}$ and $t_j$, i.e., they do not measure (which we denote below by $\mathcal{I}_{j-2:1}$). Unlike in quantum mechanics, this conditional probability can equivalently be expressed by marginalising the full joint probability distribution [see Eq.~\eqref{eq:qrf-classicalmarginalisation}] as follows
\begin{align}
    \mathbbm{P}(x_j|x_{j-1}) &= \frac{\mathbbm{P}(x_j,x_{j-1},\mathcal{I}_{j-2:1})}{\mathbbm{P}(x_{j-1},\mathcal{I}_{j-2:1})} \notag \\ 
    &=\frac{\sum_{x_{j-2},\hdots,x_1}\mathbbm{P}(x_j,x_{j-1},x_{j-2},\hdots,x_1)}{\sum_{x_{j-2},\hdots,x_1}\mathbbm{P}(x_{j-1},x_{j-2},\hdots,x_1)} \notag \\
    &= \frac{\sum_{x_{j-2},\hdots,x_1}\bra{x_j} S_{j:j-1} \ketbra{x_{j-1}}{x_{j-1}} S_{j-1:j-2} \ketbra{x_{j-2}}{x_{j-2}} \hdots \ketbra{x_2}{x_2}S_{2:1}\ketbra{x_1}{x_1}p_1\ket{x_0}}{\sum_{x_{j-2},\hdots,x_1}\bra{x_{j-1}} S_{j-1:j-2} \ketbra{x_{j-2}}{x_{j-2}} \hdots \ketbra{x_2}{x_2}S_{2:1}\ketbra{x_1}{x_1}p_1\ket{x_0}} \notag \\
    &= \bra{x_j} S_{j:j-1} \ket{x_{j-1}}.
\end{align}
Thus we have that for a memoryless classical dynamics, the conditional probabilities $\mathbbm{P}(x_j|x_{j-1},\hdots,x_1)$ and $\mathbbm{P}(x_j|x_{j-1})$ coincide (and both amount to $\bra{x_j} S_{j:j-1} \ket{x_{j-1}}$), leading to Markovianity of the statistics and consequently the decomposition of the joint probability distribution expressed in Eq.~\eqref{eq:qrf-classicalmarkov}.

Conversely, any Markovian statistics can be faithfully reproduced via a memoryless classical model: Given a joint probability distribution over measurement outcomes for a classical stochastic process, one can always write
\begin{align}\label{eq:qrf-classicaldynamicaljoint}
    \mathbbm{P}(x_n,\hdots, x_1) =\mathbbm{P}(x_n|x_{n-1},\hdots,x_1) \mathbbm{P}(x_{n-1}|x_{n-2},\hdots,x_1) \, \hdots \,
    \mathbbm{P}(x_2|x_1) \, \mathbbm{P}(x_1).
\end{align}
Equation~\eqref{eq:qrf-classicaldynamicaljoint} holds true for any probability distribution (by definition of conditional probabilities), with the decomposition on the r.h.s. encoding potential memory effects. For Markovian statistics, the above expression simplifies to Eq.~\eqref{eq:qrf-classicalmarkov}. Then, one can simply define a set of matrices $\{S_{j:j-1}\}$ via
\begin{align}
    \bra{x_j} S_{j:j-1} \ket{x_{j-1}} := \mathbbm{P}(x_j|x_{j-1}).
\end{align}
These matrices are stochastic (as they contain only non-negative entries and each of the columns to unity since $\sum_{x_j} \mathbbm{P}(x_j|x_{j-1}) = 1 \, \forall \, x_{j-1}$). One can also define the initial state via $\bra{x_1}\boldsymbol{p}_1 := \Pprob(x_1)$. From these objects, one can reproduce the joint statistics faithfully via the memoryless dynamical model expressed in Eq.~\eqref{eq:qrf-classicalregressionformula}. 

\subsection{Sub-Statistics of Memoryless Classical Dynamics}
\label{app::SubStatClass}

Here, we prove Cor.~\ref{cor:qrf-classicalmarkovsubstatistics} of the main text: \\ 

{\noindent\it{\bf Corollary~\ref{cor:qrf-classicalmarkovsubstatistics}.}
All sub-statistics of a memoryless classical dynamics are Markovian and the corresponding conditional probabilities are compatible.} \\ 

For the proof, consider a memoryless dynamics on $\Tcal_n = \{t_1, \hdots, t_n\}$ and an arbitrary sub-statistics where the experimenter measures at time $t_j$ and any subset of earlier times $\Gamma^{(i)}$ (with corresponding sequence of outcomes $\boldsymbol{x}_{\Gamma^{(i)}}$), where $t_j>t_i = \max(\Gamma^{(i)})$. Below, we assume both that $\Gamma^{(i)}$ does not `skip times' (e.g., it can be of the form $\{t_3, t_2, t_1\}$, but not $\{t_3,t_1\}$) and that $\min(\Gamma^{(i)}) = t_1$. These assumptions are not crucial and do not affect the generality of the results, but significantly simplify notation. We denote by $M$ the set of all times between $t_i$ and $t_j$ and by $F$ that of all times after $t_j$, with corresponding outcome sequences $\boldsymbol{x}_M$ and $\boldsymbol{x}_F$ and do-nothing operations $\Ical_M$ and $\Ical_F$, respectively. For convenience, we also introduce the do-nothing operation $\Ical_{FjM}$ for all times after $t_i$.
With this, we explicitly calculate the probability of observing $x_j$ conditioned on the sequence of previous outcomes $\boldsymbol{x}_{\Gamma^{(i)}}$ as
\begin{align}\label{eq:qrf-app-classicalsubstatistics}
    \mathbbm{P}(x_j|\boldsymbol{x}_{\Gamma^{(i)}}) &= \frac{\mathbbm{P}(\mathcal{I}_{F}, x_j, \mathcal{I}_M, x_i, \hdots , x_1)}{\mathbbm{P}(\mathcal{I}_{FjM},x_{i},\hdots, x_1)} \notag \\
    &= \frac{\sum_{\boldsymbol{x}_F \boldsymbol{x}_M} \mathbbm{P}(x_n,\hdots,x_1)}{\sum_{\boldsymbol{x}_F x_j \boldsymbol{x}_M} \mathbbm{P}(x_n,\hdots,x_1)} \notag \\
    &= \frac{\sum_{\boldsymbol{x}_F \boldsymbol{x}_M} \mathbbm{P}(x_n|x_{n-1})\hdots\mathbbm{P}(x_2|x_1)}{\sum_{\boldsymbol{x}_F x_j \boldsymbol{x}_M}\mathbbm{P}(x_n|x_{n-1})\hdots\mathbbm{P}(x_2|x_1)} \notag \\
    &= \frac{\left[\sum_{\boldsymbol{x}_F} \mathbbm{P}(x_n|x_{n-1})\hdots \mathbbm{P}(x_{j+1}|x_j)\right] \left[ \sum_{\boldsymbol{x}_M} \mathbbm{P}(x_j|x_{j-1}) \hdots \mathbbm{P}(x_{i+1}|x_{i}) \right] \left\{\mathbbm{P}(x_{i}|x_{i-1}) \hdots  \mathbbm{P}(x_2|x_1)\right\}}{\left[\sum_{\boldsymbol{x}_F x_j \boldsymbol{x}_M}\mathbbm{P}(x_n|x_{n-1})\hdots\mathbbm{P}(x_{i+1}|x_{i})\right] \left\{\mathbbm{P}(x_i|x_{i-1}) \hdots \mathbbm{P}(x_2|x_1)\right\}} \notag \\
    &= \sum_{\boldsymbol{x}_M} \mathbbm{P}(x_j|x_{j-1}) \hdots  \mathbbm{P}(x_{i+1}|x_i) =: \mathbbm{P}(x_j|x_i),
\end{align}
where in the second line we employed the marginalisation rule to compute the sub-statistics from the full process on $\Tcal_n$, in the third line we invoked the Markovianity condition (on the full statistics), in the fourth line we split the sums into independent parts, in the fifth line we used the fact that the first sum in the numerator and the sum in the denominator both evaluate to unity, and the final line only depends on $x_j$ and $x_i$ and satisfies the properties of a conditional probability distribution. Thus we see that any sub-statistics of a memoryless classical dynamics are also Markovian, i.e., $\mathbbm{P}(x_j|\boldsymbol{x}_{\Gamma^{(i)}}) = \mathbbm{P}(x_j|x_i)$ for all $t_j>t_i$. As mentioned, this reasoning also holds for more `complicated' sets $\Gamma^{(i)}$, albeit with a slightly more cumbersome notation than used in the proof above.

Moreover, the Markovian sub-statistics are compatible in the sense that it does not matter what occurred at any time prior to that of the most recent conditioning argument, i.e., $t_i$. For instance, if one computes $\mathbbm{P}(x_j|x_i, \Ical_{i-1:\ell+1},x_\ell,\Ical_{\ell-1:1})$, this should also be independent of $x_\ell$ (i.e., Markovian sub-statistics) \emph{and} equal to $\mathbbm{P}(x_j|\boldsymbol{x}_{\Gamma^{(i)}}) = \Pprob(x_j|x_i)$ computed above (i.e., compatible). This can be seen by noting that for any historic sequence (of either measuring or not at any times $t_1,\hdots,t_{i-1}$, which we denote with $(x\cup\Ical)_{i-1:1}$), the logic of Eq.~\eqref{eq:qrf-app-classicalsubstatistics} holds, since the only changes would appear in the terms in curly parentheses in the fourth line, which always cancel. Hence, we have the compatibility $\mathbbm{P}(x_j|x_{i},(x\cup\Ical)_{i-1:1}) = \mathbbm{P}(x_j|x_{i})$ for all possible combinations of measuring or not in the history leading up to time $t_i$. Again, this argument can be run in exactly the same vein for any two subsets of times $\Gamma^{(i)}$ and $\Gamma^{(i)\prime}$ satisfying $\max(\Gamma^{(i)}) = \max(\Gamma^{(i)\prime}) = t_i$, with the result that $\Pprob(x_j|\boldsymbol{x}_{\Gamma^{(i)}}) = \Pprob(x_j|\boldsymbol{x}_{\Gamma^{(i)\prime}}) = \Pprob(x_j|x_i)$ for all $t_j >t_i$.

\section{Quantum Dynamics}\label{app:qrf-quantum}

\subsection{Memoryless Quantum Dynamics and Markovian Statistics}

Here, we prove Lem.~\ref{lem:qrf-quantummemorylessimpliesmarkovianity} of the main text: \\ 

{\noindent\it{\bf Lemma~\ref{lem:qrf-quantummemorylessimpliesmarkovianity}.}
Any memoryless quantum dynamics leads to Markovian statistics (for sharp, projective measurements).} \\ 

Beginning with Eq.~\eqref{eq:qrf-quantumregressionformula}, we have that for any memoryless quantum dynamics
\begin{align}\label{eq:qrf-app-quantumregressionformula} 
    \mathbbm{P}(x_n,\hdots,x_1) = \tr{\Pcal_n^{(x_n)} \Lambda_{n:n-1}  \hdots \Lambda_{2:1} \Pcal_1^{(x_1)} \rho_1} \, ,
\end{align}
where $\{\Lambda_{j:j-1}\}$ are mutually independent CPTP maps, $\rho_1$ is an initial quantum state, and $\Pcal_j^{(x_j)}[\,\bullet\,] = \ketbra{x_j}{x_j}\bullet \ketbra{x_j}{x_j}$. The statistics up to any time $t_j$ is given by $\Pprob(\Ical_{n:j+1}, x_j,\hdots, x_1) =: \Pprob(x_j, \hdots, x_1) = \tr{\Pcal_j^{(x_j)} \Lambda_{j:j-1}  \hdots \Lambda_{2:1} \Pcal_1^{(x_1)} \rho_1}$ (this can be seen either by direct computation or by invoking causality), where $\Ical_{n:j+1}$ denotes `do-nothing' operations from $t_j$ to $t_n$. With this, computing the conditional probability for an arbitrary time $t_j$ given the entire sequence of historic outcomes up until that time explicitly gives
\begin{align}\label{eq:qrf-app-quantumfullsequencemarkov}
    \mathbbm{P}(x_j|x_{j-1},\hdots,x_1) &= \frac{\tr{\Pcal_j^{(x_j)} \Lambda_{j:j-1}  \hdots \Lambda_{2:1} \Pcal_1^{(x_1)} \rho_1}}{\tr{\Pcal_{j-1}^{(x_{j-1})} \Lambda_{j-1:j-2}  \hdots \Lambda_{2:1} \Pcal_1^{(x_1)} \rho_1}} \notag \\
    &= \frac{\mathrm{tr}\left[\sum_{\ell} \ketbra{x_j}{x_j} L_{j:j-1}^{\ell} \ketbra{x_{j-1}}{x_{j-1}} \left( \Lambda_{j-1:j-2}  \hdots \Lambda_{2:1} \Pcal_1^{(x_1)} \rho_1 \right) \ketbra{x_{j-1}}{x_{j-1}} L_{j:j-1}^{\ell \dagger} \ketbra{x_j}{x_j}\right]}{\tr{\ketbra{x_{j-1}}{x_{j-1}} \left( \Lambda_{j-1:j-2}  \hdots \Lambda_{2:1} \Pcal_1^{(x_1)} \rho_1\right) \ketbra{x_{j-1}}{x_{j-1}}}} \notag \\
    &= \frac{\sum_\ell \bra{x_j} L_{j:j-1}^\ell \ketbra{x_{j-1}}{x_{j-1}} L_{j:j-1}^{\ell \dagger} \ket{x_j} \bra{x_{j-1}} \left( \Lambda_{j-1:j-2}  \hdots \Lambda_{2:1} \Pcal_1^{(x_1)} \rho_1 \right) \ket{x_{j-1}}}{\bra{x_{j-1}} \left( \Lambda_{j-1:j-2}  \hdots \Lambda_{2:1} \Pcal_1^{(x_1)} \rho_1 \right) \ket{x_{j-1}}} \notag \\
    &= \sum_\ell \bra{x_j} L_{j:j-1}^\ell \ketbra{x_{j-1}}{x_{j-1}} L_{j:j-1}^{\ell \dagger} \ket{x_j} \notag \\
    &= \bra{x_j} \Lambda_{j:j-1} [\ketbra{x_{j-1}}{x_{j-1}}] \ket{x_j},
\end{align}
where we wrote $\Lambda_{j:j-1}[\,\bullet\,] := \sum_\ell L_{j:j-1}^\ell \bullet L_{j:j-1}^{\ell \dagger}$ in Kraus operator form in the second line, and then made use of the cyclicity of the trace and the fact that the measurements are sharp (rank-$1$) projectors in the third line (importantly, if the projectors are not rank-$1$, corresponding, e.g., to the measurement of an observable with degeneracies, then memoryless processes do \textit{not} necessarily lead to Markovian statistics~\cite{TarantoThesis, Taranto_2019A}, since in this case the state after the measurement is not fully determined by the outcome; this fact is also true in the classical setting). This expression is independent of all $x_1,\hdots,x_{j-2}$ and therefore the conditional probabilities are Markovian. We now show that it is indeed equivalent to the conditional probability $\mathbbm{P}(x_j|x_{j-1},\Ical_{j-2:1}) =: \mathbbm{P}(x_j|x_{j-1})$, where the experimenter does \textit{not} measure at all on times $t_1,\hdots,t_{j-2}$. Explicitly, we have
\begin{align}\label{eq:qrf-app-quantumconsecutivemarkov}
    \mathbbm{P}(x_j|x_{j-1}) &= \frac{\tr{\Pcal_j^{(x_j)} \Lambda_{j:j-1} \Pcal_{j-1}^{(x_{j-1})} \Lambda_{j-1:j-2} \Ical_{j-2} \hdots \Lambda_{2:1} \Ical_1 \rho_1}}{\tr{\Pcal_{j-1}^{(x_{j-1})} \Lambda_{j-1:j-2} \Ical_{j-2} \hdots \Lambda_{2:1} \Ical_{1} \rho_1}} \notag \\
    &= \frac{\tr{\sum_\ell \ketbra{x_j}{x_j} L_{j:j-1}^\ell \ketbra{x_{j-1}}{x_{j-1}} \left( \Lambda_{j-1:j-2} \Ical_{j-2} \hdots \Lambda_{2:1} \Ical_1 \rho_1 \right) \ketbra{x_{j-1}}{x_{j-1}} L_{j:j-1}^{\ell \dagger} \ketbra{x_j}{x_j}} }{\tr{\ketbra{x_{j-1}}{x_{j-1}} \left( \Lambda_{j-1:j-2} \Ical_{j-2} \hdots \Lambda_{2:1} \Ical_1 \rho_1 \right) \ketbra{x_{j-1}}{x_{j-1}}}} \notag \\
    &= \frac{\sum_\ell \bra{x_j} L_{j:j-1}^\ell \ketbra{x_{j-1}}{x_{j-1}} L_{j:j-1}^{\ell \dagger} \ket{x_j} \bra{x_{j-1}} \left( \Lambda_{j-1:j-2} \Ical_{j-2} \hdots \Lambda_{2:1} \Ical_1 \rho_1 \right) \ket{x_{j-1}}}{\bra{x_{j-1}} \left( \Lambda_{j-1:j-2} \Ical_{j-2} \hdots \Lambda_{2:1} \Ical_1 \rho_1 \right) \ket{x_{j-1}}} \notag \\
    &= \sum_\ell \bra{x_j} L_{j:j-1}^\ell \ketbra{x_{j-1}}{x_{j-1}} L_{j:j-1}^{\ell \dagger} \ket{x_j} \notag \\
    &= \bra{x_j} \Lambda_{j:j-1} [\ketbra{x_{j-1}}{x_{j-1}}] \ket{x_j}.
\end{align}
Thus, we see that the conditional statistics in both situations above coincide and are indeed Markovian $\mathbbm{P}(x_j|x_{j-1},\hdots,x_1) = \mathbbm{P}(x_j|x_{j-1},\Ical_{j-2:1}) = \mathbbm{P}(x_j|x_{j-1}) = \bra{x_j} \Lambda_{j:j-1} [\ketbra{x_{j-1}}{x_{j-1}}] \ket{x_j}$.


\subsection{Sub-Statistics of Memoryless Quantum Dynamics}
Here, we prove Lem.~\ref{lem:qrf-quantummemorylessimpliesmarkoviansubstatistics} from the main text:
\\ \\
{\noindent\it{\bf Lemma~\ref{lem:qrf-quantummemorylessimpliesmarkoviansubstatistics}.} Any memoryless quantum dynamics leads to Markovian sub-statistics (for sharp, projective measurements) that are compatible.}
\\ \\
Similar to App.~\ref{app::SubStatClass}, we will restrict the discussion again to subsets of $\Tcal_n$ of the form $\Gamma^{(i)} = \{t_1, \hdots, t_i\}$ and show that $\Pprob(x_j|\boldsymbol{x}_{\Gamma^{(i)}}) = \Pprob(x_j|x_i)$ holds for all $t_j>t_i$ and $t_i = \max(\Gamma^{(i)})$. Let $\Ical_{j-1:i+1}$ denote the `do-nothing' operation at all times between $t_i$ and $t_j$. With this, we obtain
\begin{align}
    \mathbbm{P}(x_j|\Ical_{j-1:i+1},\boldsymbol{x}_{\Gamma^{(i)}}) &= \frac{\tr{\Pcal_j^{(x_j)} \Lambda_{j:j-1} \Ical_{j-1} \Lambda_{j-1:j-2} \hdots \Ical_{i+1} \Lambda_{i+1:i} \Pcal_i^{(x_i)} \Lambda_{i:i-1} \mathcal{P}_{i-1}^{(x_{i-1})} \hdots \Pcal_1^{(x_1)} \rho_1}}{\tr{\Pcal_i^{(x_i)} \Lambda_{i:i-1} \mathcal{P}_{i-1}^{(x_{i-1})} \hdots \Pcal_1^{(x_1)} \rho_1}} \notag \\
    &= \frac{\bra{x_j} \Lambda_{j:j-1} \Ical_{j-1} \Lambda_{j-1:j-2} \hdots \Ical_{i+1} \Lambda_{i+1:i}[\ketbra{x_i}{x_i}] \ket{x_j} \bra{x_i} \Lambda_{i:i-1} \mathcal{P}_{i-1}^{(x_{i-1})} \hdots \Pcal_1^{(x_1)} \rho_1 \ket{x_i} }{\bra{x_i} \Lambda_{i:i-1} \mathcal{P}_{i-1}^{(x_{i-1})} \hdots \Pcal_1^{(x_1)} \rho_1 \ket{x_i}} \notag \\
    &= \bra{x_j} \Lambda_{j:j-1} \Ical_{j-1} \Lambda_{j-1:j-2} \hdots \Ical_{i+1} \Lambda_{i+1:i}[\ketbra{x_i}{x_i}] \ket{x_j} \quad \forall \ x_{i-1},\hdots,x_1.
\end{align}

In the case where \textit{no} measurements are made until time $t_i$, we similarly have
\begin{align}
    \mathbbm{P}(x_j|\Ical_{j-1:i+1},x_i,\Ical_{i-1:1}) &= \frac{\tr{\Pcal_j^{(x_j)} \Lambda_{j:j-1} \Ical_{j-1} \Lambda_{j-1:j-2} \hdots \Ical_{i+1} \Lambda_{i+1:i} \Pcal_i^{(x_i)} \Lambda_{i:i-1} \mathcal{I}_{i-1} \hdots \Ical_1 \rho_1}}{\tr{\Pcal_k^{(x_i)} \Lambda_{i:i-1} \Ical_{i-1} \hdots \Ical_1 \rho_1}} \notag \\
    &= \frac{\bra{x_j} \Lambda_{j:j-1} \Ical_{j-1} \Lambda_{j-1:j-2} \hdots \Ical_{i+1} \Lambda_{i+1:i}[\ketbra{x_i}{x_i}] \ket{x_j} \bra{x_i} \Lambda_{i:i-1} \Ical_{i-1} \hdots \Ical_1 \rho_1 \ket{x_i} }{\bra{x_i} \Lambda_{i:i-1} \Ical_{i-1} \hdots \Ical_1 \rho_1 \ket{x_i}} \notag \\
    &= \bra{x_j} \Lambda_{j:j-1} \Ical_{j-1} \Lambda_{j-1:j-2} \hdots \Ical_{i+1} \Lambda_{i+1:i}[\ketbra{x_i}{x_i}] \ket{x_j}. 
\end{align}

Thus, we see that both conditional probabilities are equal and independent of all measurement outcomes prior to $t_i$ i.e., we have $\mathbbm{P}(x_j|\Ical_{j-1:i+1},x_i,x_{i-1},\hdots,x_1)=\mathbbm{P}(x_j|\Ical_{j-1:i+1},x_i,\Ical_{i-1:1})=:\mathbbm{P}(x_j|x_i)$ and the sub-statistics are indeed Markovian. Regarding compatibility, note that for any combination of measuring or not in the times prior to $t_i$, the only changes to the above expressions occur in the numerator term that always cancels with the corresponding part in the denominator, and so compatibility also holds true. In other words, we have $\mathbbm{P}(x_j|\Ical_{j-1:i+1},x_i,(x\cup\Ical)_{i-1:1})=\mathbbm{P}(x_j|x_i)$ for all possible choices of $(x\cup\Ical)_{i-1:1}$, i.e., all possible choices of measuring or not at times $\{t_1, \hdots, t_{i-1}\}$ in the history. As for the classical case we demonstrated in App.~\ref{app::SubStatClass}, the argument above can be run in exactly the same way for more `complicated' subsets $\Gamma^{(i)}\subset \Tcal_n$, with the only difference being that the notation becomes slightly more cumbersome.

\section{Hidden Quantum Memory and Incompatibility}\label{app:qrf-hiddenquantummemory}

\begin{figure}[t]
    \centering
    \includegraphics[width=0.7\linewidth]{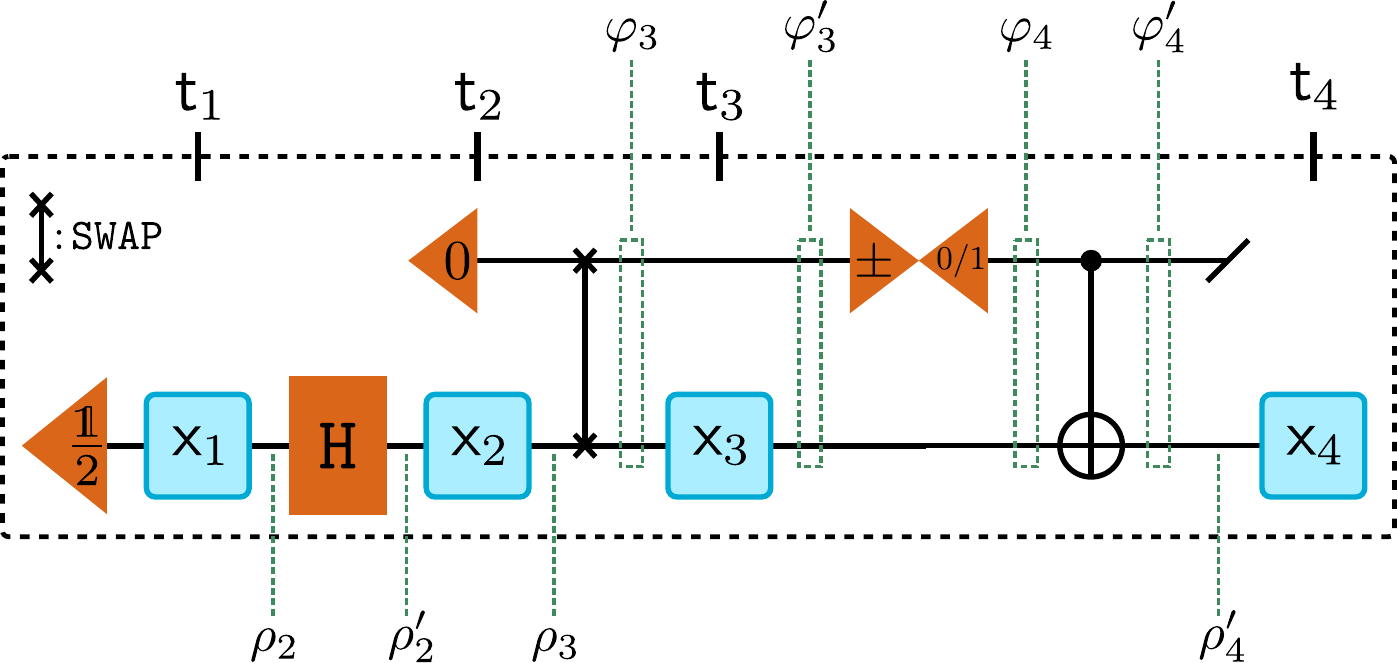}
    \caption{\textit{Circuit with hidden quantum memory.} For convenience, we reproduce the circuit provided in Fig.~\ref{fig::Circuit} in the main text. Additionally, to better facilitate orientation, the states that are explicitly mentioned throughout the proof are marked in green, i.e., the points in the circuit where the states $\rho_2, \rho_2', \rho_3, \varphi_3,\hdots$ occur.}
    \label{fig::Circuit_annotated}
\end{figure}

\subsection{Hidden Quantum Memory}

Here we explicitly calculate all sub-statistics of the example used regarding Thm.~\ref{thm:qrf-hiddennonmarkovianity} and show that, while the full statistics is Markovian, there are non-Markovian sub-statistics, i.e., we uncover hidden quantum memory. This phenomenon acts as a witness to the impossibility of a memoryless quantum dynamical model by way of contradiction with (the first part of) Lem.~\ref{lem:qrf-quantummemorylessimpliesmarkoviansubstatistics}, which states that any memoryless dynamics leads to Markovian sub-statistics. The circuit corresponding to the process we discuss is shown in Fig.~\ref{fig::Circuit_annotated}, where, for convenience, the states we explicitly calculate in the discussion below are annotated.

We begin with the full statistics. The probability over measurement outcomes at time $t_1$ are set by the initial state of the process, i.e.,
\begin{align}
    \mathbbm{P}(x_1) = \tr{\ketbra{x_1}{x_1} \rho_1},
\end{align}
with the post-measurement (sub-normalised) state given by $\rho_2(x_1) = \mathbbm{P}(x_1) \ketbra{x_1}{x_1}$. Without loss of generality, we choose $\rho_1 = \tfrac{\mathbbm{1}}{2}$ and thereby set $\mathbbm{P}(x_1=0) = \mathbbm{P}(x_1=1) = \tfrac{1}{2}$. The process then consists of a Hadamard gate, which rotates said post-measurement state (which is diagonal in the $\sigma_z$-basis) to the $\sigma_x$-basis, and we have
\begin{align}\label{eq:qrf-app-posthadamardstate}
    \rho^\prime_2(x_1) := H \rho_2(x_1) H = \begin{cases}
        \tfrac{1}{2} \ketbra{+}{+} \quad &\textup{for} \ x_1=0 \\
        \tfrac{1}{2} \ketbra{-}{-}\quad &\textup{for} \ x_1=1.
    \end{cases}
\end{align}
The experimenter then measures the $\sigma_z$ observable again, yielding the joint two-time statistics
\begin{align}
    \mathbbm{P}(x_2,x_1) = \frac{1}{4} \quad \forall \ x_2, x_1.
\end{align}
The state after the second measurement is independent of $x_1$ and given by
\begin{align}
    \rho_3(x_2,x_1) = \frac{1}{4} \ketbra{x_2}{x_2} \quad \forall \ x_2, x_1.
\end{align}
This state is then swapped with the environment, which is prepared in an arbitrary fiducial state $\tau$, which we set as the blank state $\ket{0}$. The joint system-environment state $\varphi_3$ immediately prior to the measurement at $t_3$ is given by
\begin{align}
    \varphi_3(x_2,x_1) = \texttt{SWAP} [ \rho_3(x_2,x_1) \otimes \tau ]  = \tau \otimes \rho_3(x_2,x_1) =  \frac{1}{4} \ketbra{0}{0} \otimes \ketbra{x_2}{x_2}.
\end{align}
The experimenter then measures the system at time $t_3$, recording the probabilities
\begin{align}
    \mathbbm{P}(x_3,x_2,x_1) = \begin{cases}
        \frac{1}{4} \quad &\text{for} \ x_3=0 \\
        0 \quad &\text{for} \ x_3=1 
    \end{cases}
    \quad \forall \ x_2,x_1.
\end{align}
This distribution is Markovian, as we have the conditional probabilities
\begin{align}
    \mathbbm{P}(x_3|x_2,x_1) = \begin{cases}
        1 \quad &\text{for} \ x_3=0 \\
        0 \quad &\text{for} \ x_3=1 
    \end{cases}
    \quad \forall \ x_2,x_1,
\end{align}
which are independent of $x_1$ [in fact, the statistics are `super'-Markovian as the conditional probabilities are even independent of $x_2$, so we have $\mathbbm{P}(x_3|x_2,x_1) = \mathbbm{P}(x_3|x_2) = \mathbbm{P}(x_3)$]. The system-environment state $\varphi^\prime_3(x_3,x_2,x_1)$ following the measurement at $t_3$ is
\begin{align}
    \varphi^\prime_3(x_3,x_2,x_1) = (\Pcal_3^{(x_3)}\otimes\Ical)[ \varphi_3(x_2,x_1) ]
     = \frac{1}{4} \ketbra{0}{0} \otimes \ketbra{x_2}{x_2},
\end{align}
i.e., the measurement at $t_3$ is non-invasive (note that the outcome $x_3=1$ cannot occur). Subsequently, a channel occurs that measures the environment in the $\sigma_x$-basis and feeds forward $\ket{0}$ ($\ket{1}$) whenever the measurement outcome is $+$ ($-$). The corresponding CPTP map is given by $\Upsilon[\,\bullet\,] = \sum_{k} Y^{k} \, \bullet \, Y^{k \dagger}$ with Kraus operators $Y^{0} = \ketbra{0}{+}$ and $Y^{1} = \ketbra{1}{-}$. Since $\braket{\pm | x_2} \braket{x_2 | \pm} = \tfrac{1}{2} \, \forall \, x_2$, this yields the system-environment state
\begin{align}
    \varphi_4(x_3,x_2,x_1) = (\Ical \otimes \Upsilon) [\varphi^\prime_3(x_3,x_2,x_1)] = \frac{1}{8} \ketbra{0}{0} \otimes \mathbbm{1}.
\end{align}
After this, a \verb!CNOT! gate on system and environment occurs (with the environment qubit acting as control), leading to
\begin{align}\label{eq:qrf-app-preparityfull}
    \varphi^\prime_4(x_3,x_2,x_1) = \texttt{CNOT} [\varphi_4(x_3,x_2,x_1)] = \frac{1}{8} (\ketbra{00}{00}+\ketbra{11}{11}).
\end{align}
The experimenter performs the final measurement at $t_4$, recording the probabilities
\begin{align}\label{eq:qrf-app-markovfullstatistics}
    \mathbbm{P}(x_4,x_3,x_2,x_1) = \begin{cases}
    \frac{1}{8} \quad &\text{for} \ x_3=0 \\
    0 \quad &\text{for} \ x_3=1
    \end{cases} \quad \forall \ x_4,x_2,x_1.
\end{align}
This distribution is indeed Markovian, as the conditional probabilities are
\begin{align}\label{eq:qrf-app-markovfullstatisticsconditional}
    \mathbbm{P}(x_4|x_3,x_2,x_1) = \begin{cases}
    \frac{1}{2} \quad &\textup{for} \ x_3=0 \\
    0 \quad &\textup{for} \ x_3=1 
    \end{cases} \quad \forall \ x_4, x_2,x_1,
\end{align}
[where we take the convention that conditioning on an event that cannot occur (i.e., $x_3=1$) gives conditional probability $0$]. Thus, the full joint statistics $\mathbbm{P}(x_4,x_3,x_2,x_1)$ is Markovian.

On the other hand, consider the situation in which the experimenter does not measure at time $t_2$, i.e., observes the sub-statistics $\mathbbm{P}(x_4,x_3,\Ical_2,x_1)$. Everything until Eq.~\eqref{eq:qrf-app-posthadamardstate} remains the same, but without measurement at $t_2$ we have the state
\begin{align}
    \rho_3(\Ical_2,x_1) = \mathcal{I}_2 [\rho_2^\prime(x_1)] = \frac{1}{2} \begin{cases}
        \ketbra{+}{+} \quad &\text{for} \ x_1 = 0 \\
        \ketbra{-}{-} \quad &\text{for} \ x_1 = 1.
    \end{cases}
\end{align}
The system is then swapped with the environment, yielding the joint state
\begin{align}
    \varphi_3(\Ical_2,x_1) = \texttt{SWAP} [ \rho_3(\Ical_2,x_1) \otimes \tau ]  = \tau \otimes \rho_3(\Ical_2,x_1) = \frac{1}{2} \begin{cases}
        \ketbra{0}{0} \otimes \ketbra{+}{+} \quad &\text{for} \ x_1 = 0 \\
        \ketbra{0}{0} \otimes \ketbra{-}{-} \quad &\text{for} \ x_1 = 1.
    \end{cases}
\end{align}
Measurement of the system at $t_3$ leads to the joint statistics
\begin{align}
    \mathbbm{P}(x_3,\Ical_2,x_1) = \begin{cases}
        \frac{1}{2} \quad &\text{for} \ x_3=0 \\
        0 \quad &\text{for} \ x_3=1 
    \end{cases}
    \quad \forall \ x_1.
\end{align}
Thus, we have the conditional probabilities
\begin{align}
    \mathbbm{P}(x_3|\Ical_2,x_1) = \begin{cases}
        1 \quad &\text{for} \ x_3=0 \\
        0 \quad &\text{for} \ x_3=1 
    \end{cases}
    \quad \forall \ x_1.
\end{align}
The system-environment state $\varphi^\prime_3(x_3,\Ical_2,x_1)$ following the measurement at $t_3$ is
\begin{align}
    \varphi^\prime_3(x_3,\Ical_2,x_1) = (\Pcal_3^{(x_3)}\otimes\Ical)[ \varphi_3(\Ical_2,x_1) ]
     = \frac{1}{2} \begin{cases}
        \ketbra{0}{0} \otimes \ketbra{+}{+} \quad &\text{for} \ x_1 = 0 \\
        \ketbra{0}{0} \otimes \ketbra{-}{-} \quad &\text{for} \ x_1 = 1.
    \end{cases}
\end{align}
After the map $\Upsilon$ on the environment, we have the system-environment state
\begin{align}
    \varphi_4(x_3,\Ical_2,x_1) = (\Ical \otimes \Upsilon) [\varphi^\prime_3(x_3,\Ical_2,x_1)] &= \frac{1}{2} \begin{cases}
        \ketbra{0}{0} \otimes \ketbra{0}{0} \quad &\text{for} \ x_1 = 0 \\
        \ketbra{0}{0} \otimes \ketbra{1}{1} \quad &\text{for} \ x_1 = 1
        \end{cases} \notag \\
        &= \frac{1}{2} \ketbra{0}{0} \otimes \ketbra{x_1}{x_1}.
\end{align}
Upon application of the \verb!CNOT! gate, the system-environment state is
\begin{align}\label{eq:qrf-app-preparitysub}
    \varphi^\prime_4(x_3,\Ical_2,x_1) = \texttt{CNOT} [\varphi_4(x_3,\Ical_2,x_1)] = \frac{1}{2} \ketbra{x_1}{x_1} \otimes \ketbra{x_1}{x_1}.
\end{align}
The experimenter finally performs the measurement at $t_4$ on the reduced state of the system $\rho_4^\prime(x_3,\Ical_2,x_1)=\tfrac{1}{2}\ketbra{x_1}{x_1}$ (with $x_3=0$ being the only possibility), recording the statistics
\begin{align}
    \mathbbm{P}(x_4,x_3,\Ical_2,x_1) = \begin{cases}
    \frac{1}{2} \delta_{x_4 x_1} \quad &\text{for} \ x_3=0 \\
    0 \quad &\text{for} \ x_3=1.
    \end{cases}
\end{align}
This sub-statistics is, however, non-Markovian, since the conditional probability at time $t_4$ depends on $x_1$. Explicitly, we have
\begin{align}
    \mathbbm{P}(x_4|x_3,\Ical_2,x_1) = \frac{\mathbbm{P}(x_4,x_3,\Ical_2,x_1)}{\mathbbm{P}(x_3,\Ical_2,x_1)} = \begin{cases} 
    \delta_{x_4 x_1} \quad &\textup{for} \ x_3=0 \\
    0 \quad &\textup{for} \ x_3=1
    \end{cases} \quad 
    \neq \mathbbm{P}(x_4|x_3).
\end{align}
As we have discussed in the main text, such non-Markovian sub-statistics cannot arise for a memoryless quantum dynamics probed by sharp, projective measurements (as is the case in this example), and therefore we conclude that the statistics observed---although Markovian on the whole---cannot be faithfully reproduced by a memoryless quantum dynamical model.

\subsection{Incompatibility}

\begin{figure}[t]
    \centering
    \includegraphics[width=0.7\linewidth]{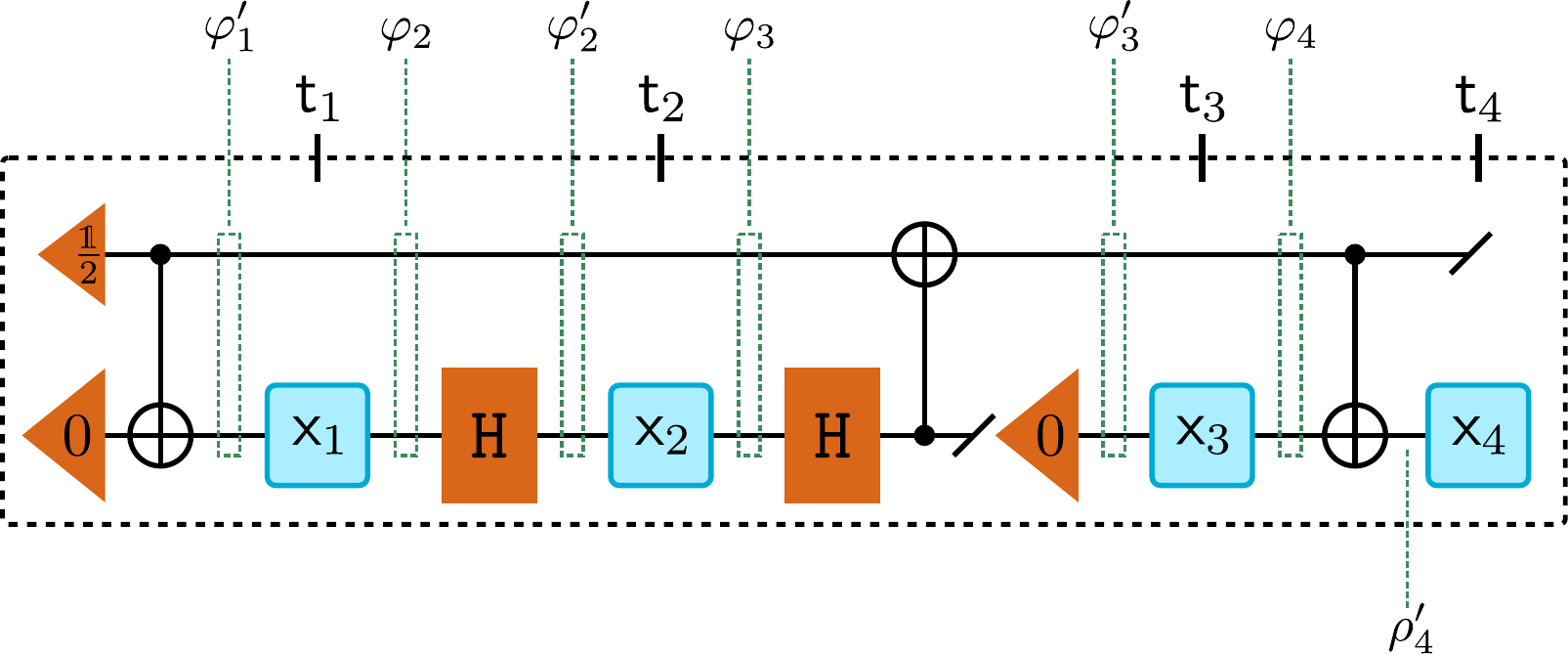}
    \caption{\textit{Circuit with incompatible Markovian (sub-)statistics.} For convenience, we reproduce the circuit provided in Fig.~\ref{fig::Circuit2} in the main text. Additionally, to better facilitate orientation, the states that are explicitly mentioned throughout the proof are marked in green, i.e., the points in the circuit where the states $\varphi_1', \varphi_2, \dots$ occur.}
    \label{fig::Circuit2_annotated}
\end{figure}

Here we explicitly calculate all sub-statistics of the example used regarding Thm.~\ref{thm:qrf-quantummarkovincompatible} and show that although they are all Markovian (i.e., unlike in Thm.~\ref{thm:qrf-hiddennonmarkovianity}, there is no explicit activation of hidden quantum memory witnessed via non-Markovian sub-statistics), they are nonetheless incompatible and therefore serve to witness the impossibility of a memoryless quantum dynamical description by way of contradiction with (the second part of) Lem.~\ref{lem:qrf-quantummemorylessimpliesmarkoviansubstatistics}.

Intuitively, the corresponding circuit is such that the state of the environment at time $t_3$, where system and environment are in a product state, does \textit{not} depend upon any previously observed measurement outcomes at times $t_1$ and $t_2$ (thus yielding Markovian (sub-)statistics for all conceivable subsets $\Gamma \subseteq \{t_1,t_2,t_3,t_4\}$ at which measurements can be performed), but rather only on \emph{whether or not} any such prior measurements were performed (thus leading to Markovian but incompatible \mbox{(sub-)}statistics. To see this explicitly, we now calculate the intermediate system-environment states at all relevant points throughout the circuit (see Fig.~\ref{fig::Circuit2_annotated} for better orientation). 

Firstly, we have 
\begin{gather} 
\varphi_1' = \frac{1}{2}(\ketbra{00}{00} + \ketbra{11}{11})\, .
\end{gather}
After time $t_1$, this state is transformed to either
\begin{gather}
\varphi_2(x_1) = \frac{1}{2} \ketbra{x_1x_1}{x_1x_1} \quad \text{or} \quad \varphi_2(\Ical_1) = \varphi_1' 
\end{gather}
depending on whether a measurement was performed or not. Following the first Hadamard gate applied to the system, we then have either
\begin{gather}
    \varphi_2^\prime(x_1) = \frac{1}{2} \ketbra{\pm}{\pm}^{x_1}\otimes\ketbra{x_1}{x_1} \quad \text{or} \quad \varphi_2^\prime(\mathcal{I}_1) = \frac{1}{2} (\ketbra{+0}{+0} + \ketbra{-1}{-1}).
\end{gather}
Here, $\ketbra{\pm}{\pm}^{x_1}$ is used to denote the state $\ketbra{+}{+}$ for $x_1=0$ and $\ketbra{-}{-}$ for $x_1=1$. Analogously, after time $t_2$, we have the system-environment states
\begin{gather}
\varphi_3(x_2,x_1) = \frac{1}{4}\ketbra{x_2x_1}{x_2x_1}, \quad \varphi_3(\Ical_2,x_1) = \frac{1}{2} \ketbra{\pm}{\pm}^{x_1}\otimes\ketbra{x_1}{x_1}, \quad \text{or} \quad \varphi_3(x_2,\Ical_1) = \frac{1}{4} \ketbra{x_2}{x_2} \otimes \mathbbm{1}.
\end{gather}
Between times $t_2$ and $t_3$, the system and environment undergo a Hadamard gate on the system, followed by a $\texttt{CNOT}$ gate (with system acting as control)), after which the system is discarded and re-prepared in the $\ket{0}$ state. Consequently, the system-environment state immediately prior to the measurement at $t_3$ is one of
\begin{gather}
\varphi_3'(x_2,x_1) = \frac{1}{8} \ketbra{0}{0} \otimes \mathbbm{1}, \quad \varphi_3'(\Ical_2,x_1) = \frac{1}{2}\ketbra{00}{00}, \quad \text{or} \quad \varphi_3'(x_2,\Ical_1) = \frac{1}{4} \ketbra{0}{0} \otimes \mathbbm{1}\, .
\end{gather}
Consequently, the state of the environment depends on whether or not measurements were performed at $t_1$ and $t_2$ (e.g., it is proportional to $\mathbbm{1}$ if both measurements were performed, and proportional to $\ketbra{0}{0}$ if only the measurement at $t_1$ was performed), while the reduced state of the system is always equal to $\ketbra{0}{0}$ for any combination of measurements and outcomes. With this, we see directly that 
\begin{gather}
\mathbbm{P}(x_3|x_2,x_1) = \mathbbm{P}(x_3|\Ical_2,x_1) = \mathbbm{P}(x_3|x_2,\Ical_1) = \delta_{x_30}, 
\end{gather}
and therefore the three-point statistics ending at time $t_3$ are Markovian (and compatible), since the corresponding conditional probabilities do not depend on \textit{any} previous outcomes. Analogously, since the system-environment state $\varphi_4$ immediately before the final $\texttt{CNOT}$ gate depends (at most) on the outcome $x_3$, but not on outcomes $x_2$ or $x_1$, \textit{all} statistics ending at time $t_4$ are also Markovian (note that for a four-step process, the statistics ending at times $t_3$ and $t_4$ are the only ones that have to be checked with respect to Markovianity). 

However, the resulting four-point conditional probabilities are not compatible. For example, we have 
\begin{gather}
\label{eqn::phi3}
\varphi_4(x_3,x_2,x_1) = \frac{1}{8} \delta_{x_3 0} \ketbra{x_3}{x_3} \otimes \mathbbm{1} \quad \text{and} \quad \varphi_4(x_3,\Ical_2,x_1) = \frac{1}{2} \delta_{x_3 0} \ketbra{x_3 0}{x_3 0}\, ,
\end{gather}
such that following the final \texttt{CNOT}, the reduced system state is either
\begin{gather}
\rho_4'(x_3,x_2,x_1) = \frac{1}{8} \delta_{x_30} \mathbbm{1} \quad \text{or} \quad \rho_4'(x_3,\Ical_2,x_1) = \frac{1}{2} \delta_{x_30} \ketbra{x_3}{x_3},
\end{gather}
depending on whether or not a measurement was made at time $t_2$.
Consequently, we can compute the corresponding statistics in either scenario:
\begin{gather}
\Pprob(x_4,x_3,x_2,x_1) = \frac{1}{8} \delta_{x_30} \quad \text{and} \quad \Pprob(x_4,x_3,\Ical_2,x_1) = \frac{1}{2} \delta_{x_30}\delta_{x_4x_3}. 
\end{gather}

In a similar vein, from Eq.~\eqref{eqn::phi3}, we can directly compute the relevant statistics for the processes ending at time $t_3$ to be
\begin{gather}
\Pprob(x_3,x_2,x_1) = \frac{1}{4} \delta_{x_30} \quad \text{and} \quad \Pprob(x_3,\Ical_2,x_1) = \frac{1}{2} \delta_{x_30}. 
\end{gather}
Finally, combining the above two equations, we obtain
\begin{gather}
\Pprob(x_4|x_3,x_2,x_1) = \frac{1}{2} \delta_{x_30} \quad \text{and} \quad \Pprob(x_4|x_3,\Ical_2,x_1) = \delta_{x_30}\delta_{x_4x_3}. 
\end{gather}
Although both of these conditional probability distributions are Markovian (i.e., show no dependence on outcomes prior to $x_3$), they nonetheless differ depending upon whether some intervention was made overall in the past, i.e., they are incompatible. Since these conditional probabilities are incompatible, there cannot be a memoryless quantum dynamics that faithfully reproduces them (as demonstrated in Lem.~\ref{lem:qrf-quantummemorylessimpliesmarkoviansubstatistics}), proving the claim.

\end{document}